\documentclass[
  preprint, onecolumn, showkeys, tightenlines,
  amsmath,amssymb,
  aps,physrev,
]{revtex4-2}

\usepackage[T1]{fontenc}
\usepackage[utf8]{inputenc}
\usepackage{bm}
\usepackage{microtype}
\usepackage{mathtools}
\usepackage{amsthm}
\usepackage{graphicx}
\usepackage{dcolumn}
\usepackage[hidelinks]{hyperref}
\usepackage[capitalize,nameinlink]{cleveref}

\makeatletter
\def\frontmatter@above@abstract{\vspace*{0pt}}
\def\frontmatter@below@abstract{\vspace*{0pt}}
\def\frontmatter@above@affiliation{\vspace*{0pt}}
\makeatother

\theoremstyle{plain}
\newtheorem{theorem}{Theorem}
\newtheorem{lemma}{Lemma}
\newtheorem{proposition}{Proposition}

\theoremstyle{definition}
\newtheorem{definition}{Definition}
\theoremstyle{remark}
\newtheorem{remark}{Remark}

\newcommand{\Assemblage}{\{\sigma_{a|x}\}}
\DeclareMathOperator{\tr}{Tr}

\begin{document}

\title{\textbf{EPR Revisited: Context-Indexed Elements of Reality and Operational Completeness}}

\author{Miko{\l}aj Sienicki}
\affiliation{Polish-Japanese Academy of Information Technology, ul.\ Koszykowa 86, 02-008 Warsaw, Poland, European Union}

\author{Krzysztof Sienicki}
\affiliation{Chair of Theoretical Physics of Naturally Intelligent Systems (\(\mathbb{NIS}\)\textsuperscript{\copyright}), Lipowa 2/Topolowa 19, 05-807 Podkowa Le\'sna, Poland, European Union}

\date{\today}

\begin{abstract}
We revisit the classic Einstein--Podolsky--Rosen (EPR) argument in a way that stays close to operational language and avoids the familiar overreach: moving from perfect predictability to the claim that incompatible observables must have simultaneous, context-independent values. The basic move is to replace EPR’s “elements of reality” by \emph{context-indexed conditional states}. These conditional states form a measurement assemblage, denoted \(\Assemblage\). Our \emph{Revised Reality Criterion} (RRC) makes the intended meaning explicit: if Alice chooses a setting \(x\) and obtains outcome \(a\), then Bob’s system is steered to a conditional state \(\omega_{B|a,x}\) that guarantees a unit-probability prediction for the corresponding test---but only within that context \((x,a)\). We then call a theory \emph{operationally complete} (relative to a specified operational arena) if it can represent all physically accessible assemblages in that arena and correctly predict their statistics. Quantum theory satisfies this requirement for the set of quantum-reachable assemblages via the Born rule and standard update rules. More generally, we show that any one-sided no-signalling operational theory supports the inference of context-indexed elements from EPR-style perfect predictions, while the stronger EPR step to \emph{simultaneous, context-free predetermined values} corresponds to adding a \emph{deterministic} local-hidden-state assumption and therefore is not forced by predictability alone. To keep the discussion concrete, we include complete proofs, a qubit singlet case study (2--3 Pauli settings with CJWR steering inequalities and noise thresholds), a continuous-variable example based on finite-squeezing Reid criteria, and a PR-box--type one-sided assemblage illustrating where quantum reachability fails.
\end{abstract}


\maketitle

\section{Introduction}
The 1935 paper by Einstein, Podolsky, and Rosen (EPR) put a sharp tension on the table: quantum mechanics seems able to predict distant outcomes with certainty, yet it does not assign simultaneous definite values to incompatible observables. The standard EPR move is to treat perfect predictability as evidence for \emph{context-free} “elements of reality.” Our aim is to keep what is operationally solid in the EPR reasoning while avoiding that extra step. Concretely, we (i) replace the original EPR phrasing by a context-indexed criterion, (ii) work in an explicitly operational framework that connects directly to experiments, and (iii) keep the presentation compatible with standard quantum mechanics.

\paragraph{Contributions.}
(1) We introduce a \emph{Revised Reality Criterion} (RRC) that makes “elements of reality” explicitly \emph{context-indexed}. (2) We define \emph{operational completeness} relative to a fixed operational arena and show how quantum theory meets this standard on the quantum-reachable subset of assemblages. (3) We formalize the extra EPR step---\emph{simultaneous predetermined values for a family of Bob tests}---as a \emph{deterministic} local-hidden-state (LHS) requirement, and we show that it is not implied by perfect predictability. (4) We connect the discussion to experiment via steering inequalities, including a qubit case study with two and three Pauli settings and realistic noise thresholds, as well as a continuous-variable illustration at finite squeezing. (5) We use the standard equivalence between measurement incompatibility and the possibility of steering to organize the main claims.

\section{Related work}
The EPR thought experiment and Bohr’s response frame the early debate about what quantum theory does—and does not—say about physical properties \cite{EPR1935,Bohr1935}. Bell’s theorem later gave this tension experimental bite by showing that any local, hidden-variable account must satisfy inequalities that quantum correlations can violate \cite{Bell1964}. Schr\"odinger introduced the term \emph{steering} for the remote preparation aspect of entanglement \cite{Schrodinger1935}. In modern operational language, steering is formulated in terms of assemblages and local-hidden-state (LHS) models \cite{Wiseman2007}, and it is now known to be tightly connected to measurement incompatibility \cite{Uola2014,Quintino2014}. For broader overviews, see \cite{CavalcantiSkrzypczyk2017,RevModPhys.92.015001}; for an accessible entry point into incompatibility as a resource, see \cite{Heinosaari2016}.

Alongside steering-based approaches, Auff\`eves and Grangier’s Contexts--Systems--Modalities (CSM) framework develops a closely related message: definite physical properties (“modalities”) are attributed to a \emph{system in a context}, not to the system alone \cite{AuffevesGrangier2016CSM}. In a complementary direction, Grangier emphasizes that the EPR dilemma can be handled by making the context explicit in the inference step, i.e., by treating context-free value attribution as the real source of the apparent conflict \cite{Grangier2021Entropy}. Our Revised Reality Criterion fits naturally with this viewpoint: it captures the same “no context-free values” lesson, but expresses it at the operational, assemblage level, and it continues to apply in GPT (Generalized Probabilistic Theories)-style settings \cite{Barrett2007GPT,Plavala2023GPT} (including non-quantum examples such as PR boxes \cite{PopescuRohrlich1994}).

\section{Operational framework}\label{sec:operational}
We adopt a general, GPT-like operational framework in which quantum mechanics appears as a special case \cite{Barrett2007GPT,Plavala2023GPT}. Systems are described by state spaces, effect spaces, and instruments (measurements together with state-update rules). We consider two agents, Alice (A) and Bob (B), who share a bipartite preparation.

Let \(\rho_{AB}\) denote a joint state shared by Alice and Bob. Alice's measurement choices are labeled by \(x\), with outcomes \(a\), and her effects are written \(M_{a|x}\) (POVM elements in QM). Given an outcome \((a,x)\), Bob's \emph{normalized} conditional state is \(\omega_{B|a,x}\). The associated unnormalized state is
\begin{equation}
  \sigma_{a|x} := p(a|x)\,\omega_{B|a,x}, \qquad p(a|x) = \tr[\sigma_{a|x}], \label{eq:sigma_def}
\end{equation}
and the assemblage \(\Assemblage\) collects all \(\sigma_{a|x}\) across Alice’s settings. Bob’s effects are denoted \(\{e_{b|y}\}\); in the quantum case, these correspond to operators \(E_{b|y}\), with evaluation rule \(e(\omega)=\tr[E\,\omega]\).

For POVMs with multiple outcomes, we say there is \emph{certainty} for a fixed setting \(y\) when exactly one outcome has probability \(1\) and the others have probability \(0\). This notion is inherently \emph{per-POVM} (see \cref{def:value}). It does not rule out certainty for effects from different POVMs when the state lies in the intersection of the relevant \(1\)-eigenspaces.

Given any POVM element \(0 \le E \le I\), the supremum of \(\tr[E\rho]\) over all states \(\rho\) is the operator norm \(\lVert E \rVert_\infty \le 1\). If a density operator \(\rho\) satisfies \(\tr[E\rho]=1\), then \(\lVert E \rVert_\infty=1\) and \(\rho\) is supported on the eigenspace of \(E\) with eigenvalue \(1\) (when such an eigenspace exists). In infinite-dimensional settings, one may have \(\lVert E \rVert_\infty=1\) without an eigenvector at eigenvalue \(1\) (continuous spectrum at \(1\)); in that case, no normal state achieves unit probability.

\begin{definition}[Assemblage]\label{def:assemblage}
Suppose Alice and Bob share a bipartite preparation. When Alice performs a measurement labeled by \(x\) and obtains outcome \(a\), Bob's unnormalized conditional state is \(\sigma_{a|x} := p(a|x)\,\omega_{B|a,x}\). Operational no-signalling from Alice to Bob requires that the marginal \(\sum_a \sigma_{a|x}\) does not depend on \(x\).
\end{definition}

\begin{definition}[Local-Hidden-State (LHS) model]\label{def:LHS}
An assemblage \(\Assemblage\) admits an LHS model if there exists a hidden variable \(\lambda\) with distribution \(p(\lambda)\), response functions \(p(a|x,\lambda)\) with \(\sum_a p(a|x,\lambda)=1\), and local normalized states \(\omega_\lambda\) for Bob such that
\begin{equation}
  \sigma_{a|x} = \sum_{\lambda} p(\lambda)\,p(a|x,\lambda)\,\omega_\lambda
  \quad \forall(a,x). \label{eq:LHS}
\end{equation}
\end{definition}

\begin{definition}[Revised Reality Criterion (RRC)]\label{def:RRC}
Whenever Alice’s outcome enables a perfect, unit-probability prediction for some test on Bob’s side, we say that Bob’s system carries a \emph{context-indexed element of reality}, namely the conditional state \(\omega_{B|a,x}\) tied to the specific context \((x,a)\).
\end{definition}

\begin{definition}[Operational completeness]\label{def:opcomp}
Let \(\mathcal{A}\) be an \emph{operational arena}—a specified set of joint preparations and instruments on systems \(A\) and \(B\). A theory is \emph{operationally complete (relative to \(\mathcal{A}\))} if, for every assemblage \(\Assemblage\) generated within \(\mathcal{A}\) that satisfies one-sided no-signalling, the theory provides a representation of the conditional states and correct predictions for all associated measurement statistics. We assume \(\mathcal{A}\) is closed under classical mixing and classical post-processing.
\end{definition}

\begin{proposition}[Closure and scope; counterexample]\label{prop:closure_scope}
(i) If \(\mathcal{A}\) is closed under classical mixing and post-processing, then the set of \(\mathcal{A}\)-reachable assemblages is convex and stable under deterministic coarse-grainings on Alice's side. (ii) Nonetheless, there exist one-sided no-signalling assemblages that are not quantum-realizable. A sharp example is a PR-box--type one-sided assemblage \cite{PopescuRohrlich1994} (specified operationally in Appendix~\ref{app:prbox}): it is one-sided no-signalling yet yields CHSH value \(4\), exceeding the Tsirelson bound \(2\sqrt{2}\).
\end{proposition}

When \(p(a|x)=0\), the normalized conditional state \(\omega_{B|a,x}\) is formally undefined. Expressions involving \(e_{b|y}(\omega_{B|a,x})\) are understood to apply only to contexts with \(p(a|x)>0\). The unnormalized state \(\sigma_{a|x}\) remains well-defined in all cases.

\begin{definition}[Context-free value assignment]\label{def:value}
A \emph{context-free value assignment} is a deterministic map \(v:\{e_{b|y}\}\to\{0,1\}\) such that for each POVM \(\{e_{b|y}\}_b\) at fixed \(y\),
\begin{equation}
\sum_b v(e_{b|y}) = 1, \qquad
v\!\left(\sum_{b \in B} e_{b|y}\right) = \sum_{b \in B} v(e_{b|y}). \label{eq:value_consistency}
\end{equation}
That is, exactly one effect per POVM is assigned value \(1\), the rest \(0\), and this extends consistently to coarse-grained effects.
\end{definition}

\begin{remark}[What the EPR ``predetermined values'' move actually assumes]
EPR-style ``simultaneous elements of reality'' do \emph{not} require a single fixed value map \(v\) that is the same in every run.
Rather, they amount to assuming that \emph{in each run} there exists some underlying state (hidden variable) \(\lambda\) that determines outcomes for \emph{all} of Bob's relevant tests, while Alice's choice/outcome merely conditions which \(\lambda\)'s are compatible with the run.
In this note we make this precise in a \emph{strong (deterministic)} form, matching the usual ``predetermined outcomes'' reading under perfect predictability; see \cref{def:PV}.
\end{remark}

\begin{definition}[Predetermined-values (PV) model for Bob (strong/deterministic)]\label{def:PV}
Fix a family of Bob tests \(\mathcal{Y}\) (e.g.\ two incompatible observables).
An assemblage \(\Assemblage\) admits a \emph{PV model for Bob on \(\mathcal{Y}\)} if it admits an LHS decomposition \eqref{eq:LHS} in which, for each hidden index \(\lambda\) and each test \(y\in\mathcal{Y}\), there is a \emph{deterministic} outcome \(b_y(\lambda)\) such that
\begin{equation}
e_{b_y(\lambda)|y}(\omega_\lambda)=1
\qquad\text{and hence}\qquad
e_{b|y}(\omega_\lambda)\in\{0,1\}\ \ \forall b.
\label{eq:PV_determinism}
\end{equation}
In words: each \(\lambda\) carries simultaneous, context-independent predetermined outcomes for all tests in \(\mathcal{Y}\), while Alice's context \((x,a)\) only changes the \emph{classical} weights over \(\lambda\) via Bayes conditioning through \(p(a|x,\lambda)\).
\end{definition}

\section{Main results}\label{sec:main}

\paragraph{Assumptions.}
We assume: (i) one-sided operational no-signalling from A to B; (ii) well-defined conditional states \(\omega_{B|a,x}\) for all instruments on A; and (iii) operational completeness within a fixed arena \(\mathcal{A}\).
In the quantum case, measurement effects are POVM elements and Alice’s instruments are taken to be \emph{local} completely positive, trace-nonincreasing maps acting on \(A\) alone (not involving \(B\)); state updates need not be of L\"uders type.
Under measurement independence, parameter independence on Bob’s side implies one-sided no-signalling when averaged over hidden variables.

\begin{lemma}[EPR predetermined-values assumption \(\Rightarrow\) deterministic LHS]\label{lem:parent}
If an assemblage \(\Assemblage\) admits a PV model for Bob on a family of tests \(\mathcal{Y}\) (Definition~\ref{def:PV}),
then \(\Assemblage\) is unsteerable on \(\mathcal{Y}\), i.e.\ it admits an LHS model reproducing Bob's statistics for all \(y\in\mathcal{Y}\).
Consequently, any assemblage that is steerable on \(\mathcal{Y}\) admits \emph{no} PV model on \(\mathcal{Y}\).
\end{lemma}

\begin{proof}
By definition, a PV model includes an LHS decomposition \eqref{eq:LHS}; the additional condition \eqref{eq:PV_determinism} makes Bob's responses on \(\mathcal{Y}\) deterministic at the hidden-state level.
Thus PV is a special case of (unsteerable) LHS structure on \(\mathcal{Y}\).
The contrapositive gives the final statement.
\end{proof}

\begin{lemma}[Partial-trace push-through identity]\label{lem:pushthrough}
For any operator \(K\) on Alice’s system and any bipartite state \(\rho_{AB}\),
\begin{equation}
\tr_A\!\big[(K\!\otimes\! I)\,\rho_{AB}\,(K^\dagger\!\otimes\! I)\big]
= \tr_A\!\big[(K^\dagger K\!\otimes\! I)\,\rho_{AB}\big]
= \tr_A\!\big[\rho_{AB}\,(K^\dagger K\!\otimes\! I)\big]. \label{eq:pushthrough}
\end{equation}
\end{lemma}

\begin{proposition}[Quantum realization (local-instrument form)]\label{prop:quantum}
Suppose Alice implements a \emph{local} instrument on \(A\) with Kraus operators \(\{K_{a\mu|x}\}_\mu\) such that
\(\sum_\mu K_{a\mu|x}^\dagger K_{a\mu|x} = M_{a|x}\).
For a shared state \(\rho_{AB}\), Bob’s unnormalized conditional state is
\begin{equation}
  \sigma_{a|x}
  = \sum_\mu \tr_A\!\Big[(K_{a\mu|x}\otimes I)\,\rho_{AB}\,(K_{a\mu|x}^\dagger\otimes I)\Big]
  = \tr_A[(M_{a|x}\otimes I)\,\rho_{AB}], \label{eq:quantum_sigma}
\end{equation}
and the normalized state is \(\omega_{B|a,x}=\sigma_{a|x}/p(a|x)\) with \(p(a|x)=\tr[\sigma_{a|x}]\).
By \cref{lem:pushthrough}, \(\sigma_{a|x}\) depends only on the POVM effects \(M_{a|x}\), not on the particular Kraus decomposition \(\{K_{a\mu|x}\}\).
\end{proposition}

\begin{theorem}[No-paradox under RRC and one-sided no-signalling]\label{thm:noparadox}
In any operational theory satisfying one-sided no-signalling and capable of representing conditional states and their statistics, EPR-style perfect predictability implies the existence of context-indexed elements \(\omega_{B|a,x}\) (RRC),
but it does \emph{not} entail the stronger EPR conclusion that Bob possesses \emph{simultaneous, context-free predetermined outcomes} for a family of incompatible tests \(\mathcal{Y}\).
That stronger conclusion corresponds to assuming a PV model on \(\mathcal{Y}\) (Definition~\ref{def:PV}), hence (by Lemma~\ref{lem:parent}) to assuming an (unsteerable) LHS structure on \(\mathcal{Y}\).
\end{theorem}

\begin{proof}
Perfect predictability for a given Bob test \(y\) in a given context \((x,a)\) means there exists an outcome \(b\) such that
\(e_{b|y}(\omega_{B|a,x})=1\).
By the RRC, the operative ``element'' is the \emph{context-indexed} conditional state \(\omega_{B|a,x}\).

The extra EPR step is to claim more than this: namely, that in each run Bob carries simultaneous predetermined outcomes for an entire family \(\mathcal{Y}\) of tests, independently of which measurement Alice chose.
Operationally, that is exactly the content of a PV model on \(\mathcal{Y}\): one posits hidden states \(\omega_\lambda\) whose responses on \(\mathcal{Y}\) are deterministic, and Alice's context \((x,a)\) only updates classical weights over \(\lambda\) via \(p(a|x,\lambda)\); see Definition~\ref{def:PV}.

But PV is \emph{not} implied by perfect predictability alone. Indeed, PV is a special case of LHS structure (Lemma~\ref{lem:parent}).
\emph{Concrete counterexample (quantum):} for the singlet state, if Alice measures either \(\sigma_x\) or \(\sigma_z\) then Bob is steered (in each context \((x,a)\)) to the corresponding Pauli eigenstate, giving perfect predictability for Bob measuring the same Pauli observable, yet the resulting two-setting assemblage violates a two-setting linear steering inequality (e.g.\ CJWR) \cite{CJWR2009} and hence admits no LHS (therefore no PV) model on \(\mathcal{Y}=\{\sigma_x,\sigma_z\}\).

Therefore, whenever an assemblage is steerable on \(\mathcal{Y}\) (i.e.\ admits no LHS model reproducing the statistics for \(\mathcal{Y}\)), no PV model exists.
Hence there is no logical route from ``predictability in some contexts'' to ``simultaneous context-free predetermined values for incompatible tests'' without adding an LHS/PV assumption.
\end{proof}

\begin{theorem}[Incompatibility \(\Leftrightarrow\) steering]\label{thm:incompatibility_steering}
(\textbf{Uola--Moroder--G\"uhne}; \textbf{Quintino--V\'ertesi--Brunner}) \\
With trusted Bob (arbitrary dimension) and general POVMs on Alice’s side, a set of Alice’s measurements is jointly measurable \emph{if and only if} they cannot demonstrate steering—i.e., if and only if all assemblages generated from them admit an LHS model \cite{Uola2014,Quintino2014}.
\end{theorem}

\begin{proof}[Proof sketch]
(Intuition only.) Joint measurability provides a parent POVM and classical post-processings; inserting this structure into the assemblage definition yields an LHS model. Conversely, if no assemblage violates the LHS condition, one can reconstruct a parent POVM using convex separation arguments. Full derivations appear in \cite{Uola2014,Quintino2014}.
\end{proof}

\section{Worked qubit example: singlet, CJWR steering, and noise thresholds}\label{sec:qubit}
Consider the Werner state \(\rho_W = p\,|\Psi^-\rangle\!\langle\Psi^-| + (1-p)\,I_4/4\), where \(p\) is the mixing parameter. Let Alice choose either \(m=2\) or \(3\) Pauli settings, and let Bob measure along the corresponding directions. Define
\begin{equation}
  S_m := \frac{1}{m} \sum_{x=1}^m \langle A_x B_x \rangle, \qquad \text{(our normalization)} \label{eq:Sm}
\end{equation}
and compare with the CJWR functional \cite{CJWR2009}:
\begin{equation}
  F_m^{\rm CJWR} := \frac{1}{\sqrt{m}} \Bigl| \sum_{x=1}^m \langle A_x B_x \rangle \Bigr|
  = \sqrt{m}\,|S_m| \le 1 \quad (\text{LHS bound}). \label{eq:CJWR_relation}
\end{equation}
So the LHS bound in our normalization is
\begin{equation}
  |S_m| \le \frac{1}{\sqrt{m}}. \label{eq:Sm_LHS_correct}
\end{equation}
For the Werner state with matched settings, \(\langle A_x B_x \rangle = -p\), hence \(|S_m|=p\). Thus this \emph{specific CJWR witness with matched settings} is violated when
\begin{equation}
  p > \frac{1}{\sqrt{m}}
  \quad \Rightarrow \quad
  p > \frac{1}{\sqrt{2}}~(m=2),\qquad
  p > \frac{1}{\sqrt{3}}~(m=3). \label{eq:werner_thresholds_correct}
\end{equation}
For \(m=2\), this threshold coincides \emph{numerically} with the usual CHSH nonlocality threshold for Werner states (even though the CJWR functional is not the CHSH expression). For \(m=3\), the threshold drops to about \(0.577\), illustrating the strict inclusion \(\text{Bell} \subset \text{Steering} \subset \text{Entanglement}\).

\section{Continuous variables: finite-squeezing Reid criterion}\label{sec:cv}
We use units with \(\hbar=1\). Consider a two-mode squeezed vacuum with squeezing parameter \(r\). For a jointly Gaussian state, the optimal \emph{linear} estimator gives
\begin{equation}
  \Delta^2_{\mathrm{inf}} x_B = \Delta^2_{\mathrm{inf}} p_B
  = \frac{1}{2}\,\operatorname{sech}(2r) = \frac{1}{2\cosh(2r)}. \label{eq:cv_var}
\end{equation}
Reid’s steering criterion is
\begin{equation}
  \Delta_{\mathrm{inf}} x_B\,\Delta_{\mathrm{inf}} p_B < \frac{1}{2}
  \quad \Leftrightarrow \quad
  \Delta^2_{\mathrm{inf}} x_B\,\Delta^2_{\mathrm{inf}} p_B < \frac{1}{4}. \label{eq:reid}
\end{equation}
For a representative finite squeezing \(r=0.69\) (about 6~dB), one finds \(\Delta_{\mathrm{inf}} x_B = \Delta_{\mathrm{inf}} p_B \approx 0.486\), hence
\begin{equation}
  \Delta_{\mathrm{inf}} x_B\,\Delta_{\mathrm{inf}} p_B \approx 0.237 < \tfrac{1}{2}, \label{eq:reid_number}
\end{equation}
certifying continuous-variable EPR steering. While the linear estimator is optimal for Gaussian states, it also upper-bounds the inferred variances in many non-Gaussian scenarios.

\section{Experimental diagnostics (brief)}
In the trusted-Bob scenario, measurement incompatibility on Alice’s side is necessary and sufficient for steering with some shared state (\cref{thm:incompatibility_steering}). Steering inequalities such as CJWR (\cref{sec:qubit}) and Reid’s criterion (\cref{sec:cv}) provide practical, asymmetric tests for locating an experiment relative to the LHS boundary.

\section{Relation to other accounts}
\begin{table*}[t]
\caption{Rosetta stone for various responses to the EPR dilemma. Column~1 lists operational no-signalling assumptions, while the final column outlines how each account neutralizes the apparent paradox.}
\label{tab:accounts}
\begin{ruledtabular}
\begin{tabular}{lcccccl}
Account & No-signalling (operational) & Ontological mechanism & Single outcomes & Context-free joint values & \(\psi\) status & Resolution \\
\hline
Bohr (complementarity) & \(\checkmark\) & Phenomena/contexts & \(\checkmark\) & No & Phenomena-centric (operational) & Context defines predicates \\
CSM (Auff\`eves--Grangier) \cite{AuffevesGrangier2016CSM,Grangier2021Entropy}
 & \(\checkmark\) & Context-indexed modalities & \(\checkmark\) & No & Contextual-objective & Replace EPR ``elements'' by modalities/contexts \\
Steering (this work)   & \(\checkmark\) & Standard QM updates & \(\checkmark\) & No & \(\psi\)-agnostic (operational) & Conditional states suffice \\
Bohmian mechanics      & \(\checkmark\) & Nonlocal dynamics & \(\checkmark\) & No (measurement contextual) & \(\psi\)-ontic & Accept nonlocality \\
Everett (Many-Worlds)  & \(\checkmark\) & Local unitary & Branch-relative & n/a & \(\psi\)-ontic & Branch correlations; records are single-outcome in branches \\
GRW/CSL                & \(\checkmark\) & Stochastic collapses & \(\checkmark\) & Mixed & \(\psi\)-ontic & Objective collapses \\
QBism / Relational     & \(\checkmark\) & Agent-/relation-relative & \(\checkmark\) & No & Agent-/relation-relative & Reinterpret probabilities \\
\end{tabular}
\end{ruledtabular}
\end{table*}

\appendix

\section{Why ``predetermined values'' is an extra assumption (Lemma~\ref{lem:parent})}\label{app:parent_proof}
Lemma~\ref{lem:parent} simply records the correct logical target.
What EPR needs for ``simultaneous elements of reality'' is not a single fixed value map \(v\), but a \emph{per-run} assignment carried by a hidden state \(\lambda\) that determines outcomes for all Bob tests in \(\mathcal{Y}\).
Once expressed operationally, this is exactly the PV condition: an LHS decomposition whose hidden states \(\omega_\lambda\) yield deterministic outcomes on \(\mathcal{Y}\).
Hence, if an assemblage is steerable on \(\mathcal{Y}\), no such PV (and thus no EPR-style predetermined-values conclusion on \(\mathcal{Y}\)) can be drawn.

\section{Expanded proof of the No-Paradox Theorem (Theorem~\ref{thm:noparadox})}\label{app:noparadox_proof}
Perfect predictability supplies, for each context \((x,a)\), a conditional state \(\omega_{B|a,x}\) that yields unit probability for some test on Bob’s side.
The extra EPR claim that Bob carries simultaneous, context-free predetermined outcomes for a family \(\mathcal{Y}\) of tests is precisely the PV assumption: a deterministic LHS refinement on \(\mathcal{Y}\) (Definition~\ref{def:PV}).
Since PV is not implied by predictability and is outright impossible whenever the assemblage is steerable on \(\mathcal{Y}\) (Lemma~\ref{lem:parent}), the EPR paradox dissolves once one keeps track of the additional hidden-structure assumption.

\section{PR-box--type one-sided assemblage (explicit GPT specification)}\label{app:prbox}
Here we specify a one-sided no-signalling assemblage in a \emph{box-world GPT} (non-signalling polytope), where ``states'' are conditional probability tables and ``effects'' read out the corresponding probabilities; see, e.g., \cite{Barrett2007GPT,Plavala2023GPT}. The PR-box correlation itself is due to Popescu and Rohrlich \cite{PopescuRohrlich1994}.

\paragraph{State/effect convention.}
Bob has two measurement settings \(y\in\{0,1\}\) and binary outcomes \(b\in\{0,1\}\).
A (normalized) GPT state for Bob is the table \(\omega=\{p(b|y)\}_{b,y}\).
The effect \(e_{b|y}\) evaluates \(\omega\) by returning \(e_{b|y}(\omega)=p(b|y)\).

\paragraph{Assemblage definition.}
Let Alice have inputs \(x\in\{0,1\}\) and outcomes \(a\in\{0,1\}\), with \(p(a|x)=1/2\).
Define Bob's \emph{conditional} GPT state \(\omega_{B|a,x}\) by the PR-box rule
\begin{equation}
p(b|y; a,x) \;=\; \delta_{\,b,\; a\oplus xy},
\label{eq:prbox_rule}
\end{equation}
where \(\oplus\) is addition mod 2.
Equivalently:
for \(y=0\), Bob outputs \(b=a\) always; for \(y=1\), Bob outputs \(b=a\) if \(x=0\) and \(b=1-a\) if \(x=1\).
The subnormalized assemblage elements are
\begin{equation}
\sigma_{a|x} \;:=\; \tfrac12\,\omega_{B|a,x}.
\label{eq:prbox_sigma}
\end{equation}

\paragraph{One-sided no-signalling.}
Summing over \(a\) gives Bob's marginal (subnormalized) state
\begin{equation}
\sum_a \sigma_{a|x} \;=\; \tfrac12\sum_a \omega_{B|a,x}.
\end{equation}
For each fixed \(x\) and \(y\), the two conditional states \(\omega_{B|0,x}\) and \(\omega_{B|1,x}\) assign deterministic but opposite outcomes, so their average is uniform:
\begin{equation}
p(b|y) \;=\; \tfrac12 \quad \text{for all } b,y,
\end{equation}
independent of \(x\). Thus the assemblage is one-sided no-signalling from Alice to Bob.

\paragraph{CHSH value.}
The assemblage induces a bipartite distribution via
\begin{equation}
p(a,b|x,y) \;=\; p(a|x)\, e_{b|y}\!\big(\omega_{B|a,x}\big).
\label{eq:prbox_joint_from_assemblage}
\end{equation}
Define \(\pm 1\) observables \(A_x:=(-1)^a\) and \(B_y:=(-1)^b\).
From \eqref{eq:prbox_rule} we have \(b=a\oplus xy\), hence
\[
A_x B_y = (-1)^a\,(-1)^{a\oplus xy}=(-1)^{xy}.
\]
Therefore the correlators are \(E_{xy}=\langle A_x B_y\rangle = (-1)^{xy}\), and the usual CHSH combination is
\begin{equation}
\mathrm{CHSH} \;=\; E_{00}+E_{01}+E_{10}-E_{11}
\;=\; 1+1+1-(-1)\;=\;4,
\end{equation}
the algebraic maximum.
This exceeds the quantum Tsirelson bound \(2\sqrt{2}\), so the assemblage is not quantum-realizable even though it is one-sided no-signalling in this GPT sense.

\section{CJWR functional scaling (for readers cross-checking)}\label{app:cjwr_scaling}
The functional \(S_m\) introduced in \cref{eq:Sm} connects to the CJWR form via:
\begin{equation}
  F_m^{\rm CJWR} = \frac{1}{\sqrt{m}} \Bigl| \sum_{x=1}^m \langle A_x B_x \rangle \Bigr| = \sqrt{m} |S_m|.
\end{equation}
Thus, the LHS bound \(F_m^{\rm CJWR} \le 1\) directly translates to \(|S_m| \le 1/\sqrt{m}\), matching the thresholds derived in \cref{eq:Sm_LHS_correct} and \cref{eq:werner_thresholds_correct} for Werner states with aligned Pauli settings.

\bibliographystyle{apsrev4-2}
\bibliography{refs}

\end{document}